\documentclass[runningheads]{llncs}
\usepackage[utf8]{inputenc}
\usepackage{amssymb}
\usepackage{amsmath}
\usepackage{latexsym}
\usepackage{textcomp}
\usepackage{graphicx}
\usepackage{subfigure}
\usepackage{xspace}
\usepackage{todonotes}
\usepackage{paralist} 
\usepackage{hyperref}
\usepackage{cleveref}
\usepackage{wrapfig}
\usepackage{thm-restate}
\usepackage{lineno}
\usepackage{marvosym}

\Crefname{observation}{Observation}{Observations}
\Crefname{algorithm}{Algorithm}{Algorithms}
\Crefname{section}{Section}{Sections}
\Crefname{observation}{Observation}{Observations}
\Crefname{lemma}{Lemma}{Lemmas}
\Crefname{claim}{Claim}{Claims}
\Crefname{claimx}{Claim}{Claims}
\Crefname{figure}{Fig.}{Figs.}
\Crefname{enumi}{Property}{Properties}
\Crefname{property}{Property}{Properties}

\spnewtheorem*{sketch}{Sketch of proof}{\itshape}{\rmfamily}

\newcommand{\remove}[1]{}

\newcommand{\N}{\mathcal N_u}
\newcommand{\NB}{\mathcal N}
\newcommand{\B}[1]{\beta(#1)}
\newcommand{\C}[1]{\chi(#1)}
\newcommand{\fl}{\varphi}
\newcommand{\f}[1]{\fl(#1)}
\newcommand{\ai}[2]{a_{#2}(\hat{#1})}

\newcommand{\qedclaim}{\hfill $\blacksquare$}

\graphicspath{{figures/}}

\newif \ifArxiv
\Arxivtrue

\begin{document}
\title{Quasi-upward Planar Drawings\\ with Minimum Curve Complexity\thanks{This work is partially supported by: $(i)$ MIUR, grant 20174LF3T8 ``AHeAD: efficient Algorithms for HArnessing networked Data'', $(ii)$ Dipartimento di Ingegneria - Universit\`a degli Studi di Perugia, grants RICBA19FM: ``Modelli, algoritmi e sistemi per la visualizzazione di grafi e reti'' and RICBA20EDG: ``Algoritmi e modelli per la rappresentazione visuale di reti''.
}}


\author{
	Carla~Binucci{ \href{https://orcid.org/0000-0002-5320-9110}{\protect\includegraphics[scale=0.45]{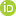}}
	}{}\textsuperscript{(\Letter)} \and
	Emilio~Di~Giacomo\texorpdfstring{ \href{https://orcid.org/0000-0002-9794-1928}{\protect\includegraphics[scale=0.45]{orcid.png}}}{} \and
	\\Giuseppe~Liotta{ \href{https://orcid.org/0000-0002-2886-9694}{\protect\includegraphics[scale=0.45]{orcid.png}}}{} \and
	Alessandra~Tappini\texorpdfstring{ \href{https://orcid.org/0000-0001-9192-2067}{\protect\includegraphics[scale=0.45]{orcid.png}}}{}
}


\institute{
Universit\`a degli Studi di Perugia, Italy\\
\email{\{carla.binucci,emilio.digiacomo,giuseppe.liotta,alessandra.tappini\}@unipg.it}
}

\maketitle

\newtheorem{observation}{Observation}
\newtheorem{claimx}{Claim}

\begin{abstract}
	
This paper studies the problem of computing quasi-upward planar drawings of bimodal plane digraphs with  minimum curve complexity, i.e., drawings such that the maximum number of bends per edge is minimized.
We prove that every bimodal plane digraph admits a quasi-upward planar drawing with curve complexity two, which is worst-case optimal. We also show that the problem of minimizing the curve complexity in a quasi-upward planar drawing can be modeled as a min-cost flow problem on a unit-capacity planar flow network. This gives rise to an $\tilde{O}(m^\frac{4}{3})$-time algorithm that computes a quasi-upward planar drawing with minimum curve complexity; in addition, the drawing has the minimum number of bends when no edge can be bent more than twice. For a contrast, we show  bimodal planar digraphs whose bend-minimum quasi-upward planar drawings require linear curve complexity even in the variable embedding setting.

\end{abstract}

\section{Introduction}
Let $G$ be a \emph{plane digraph}, i.e., a directed graph with a given planar embedding. A vertex $v$ of $G$ is \emph{bimodal} if the circular order of the edges around $v$ can be partitioned into two (possibly empty) sets of consecutive edges, one consisting of the incoming edges and the other one consisting of the outgoing edges. If every vertex of $G$ is bimodal, $G$ is a \emph{bimodal plane digraph}. See for example \Cref{fig:quasi-upward-a}.

A planar drawing of a bimodal plane digraph $G$ is \emph{upward planar} if all the edges are represented by curves monotonically increasing in the vertical direction. A digraph that admits an upward planar drawing is \emph{upward planar}. Having a bimodal embedding is a necessary but not sufficient condition for a digraph to be upward planar~\cite{dett-gd-99}. For example, the (embedded) digraph in \Cref{fig:quasi-upward-a} is not upward planar.

Garg and Tamassia~\cite{DBLP:journals/siamcomp/GargT01} proved that testing a bimodal digraph for upward planarity is NP-hard in the variable embedding setting, i.e., when all possible bimodal planar embeddings must be checked. In this setting, an  $O(n^4)$-time algorithm exists for series-parallel digraphs~\cite{DBLP:journals/siamdm/DidimoGL09}, where $n$ is the number of vertices. FPT solutions, SAT formulations, and branch-and-bound approaches have also been proposed for general digraphs (see, e.g.,~\cite{DBLP:journals/algorithmica/BertolazziBD02,10.1007/978-3-540-30140-0_16,DBLP:journals/siamdm/DidimoGL09,DBLP:journals/ijfcs/HealyL06}). On the other hand, upward planarity testing can be solved in polynomial time in the fixed embedding setting, i.e., when the input is a bimodal plane digraph $G$ and the algorithm tests whether $G$ admits an upward planar drawing that preserves the given bimodal embedding~\cite{DBLP:journals/algorithmica/BertolazziBLM94}. See also~\cite{DBLP:reference/algo/Didimo16} for a survey on upward planarity.

%

\begin{figure}[tb]
	\centering
	\subfigure[]{
		\includegraphics[width=0.3\textwidth, page=1]{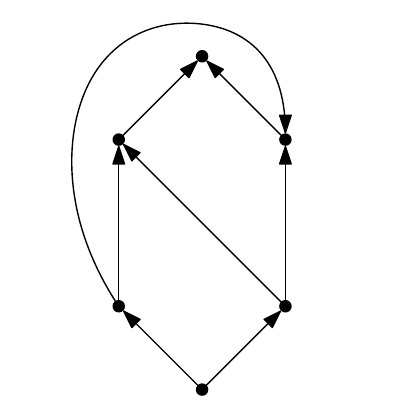}
		\label{fig:quasi-upward-a}
	}
	\hfil
	\subfigure[]{
		\includegraphics[width=0.3\textwidth, page=2]{figures/quasi-upward}
		\label{fig:quasi-upward-b}
	}
	\hfil
\subfigure[]{
	\includegraphics[width=0.3\textwidth, page=3]{figures/quasi-upward}
	\label{fig:quasi-upward-pl}
}
	\caption{(a) A bimodal plane digraph $G$. (b) A quasi-upward planar~drawing~of~$G$. (c) The same drawing with poly-line edges.}
	\label{fig:quasi-upward}
\end{figure}

Motivated by the observation that only restricted families of bimodal digraphs are upward planar, Bertolazzi et al. introduced quasi-upward planar drawings~\cite{DBLP:journals/algorithmica/BertolazziBD02}. A drawing $\Gamma$ of a digraph $G$ is \emph{quasi-upward planar} if it has no edge crossings and for each vertex $v$ there exists a sufficiently small disk $R$ of the plane, properly containing $v$, such that, in the intersection of $R$ with $\Gamma$, the horizontal line through $v$ separates the incoming edges (below the line) from the outgoing edges (above the line); see~\Cref{fig:quasi-upward-b}. Intuitively, all the incoming edges enter $v$ from ``below'' and all the outgoing edges leave $v$ from ``above''. A digraph that admits a quasi-upward planar drawing is \emph{quasi-upward planar}. An edge of a quasi-upward planar drawing that is not upward has at least one horizontal tangent. Each point of horizontal tangency is called a \emph{bend} (see~\Cref{fig:quasi-upward-b}). This term is justified by the fact that an edge with $b$ points of horizontal tangency can be represented as a poly-line with $b$ bends (substituting each point of tangency with a vertex $v$ and suitable orienting the edges incident to $v$, we obtain an upward planar digraph which always has a straight-line upward planar drawing~\cite{dett-gd-99}); see~\Cref{fig:quasi-upward-pl}. 

Bertolazzi et al.~\cite{DBLP:journals/algorithmica/BertolazziBD02} prove that, different from upward planarity, having a planar bimodal embedding is necessary and sufficient for a digraph to be quasi-upward planar.  They also study the problem of computing quasi-upward planar drawings with the minimum number of bends and use a suitable flow network to solve it in ${\tilde O}(n^{2})$-time in the fixed embedding setting. They also describe a branch-and-bound algorithm in the variable embedding setting. A list of papers about quasi-upward planarity also includes~\cite{DBLP:conf/walcom/BinucciD14,DBLP:journals/cj/BinucciD16,DBLP:journals/tcs/BinucciDP14}.
%

\medskip
\noindent\textbf{Our contribution.} In this paper we study the problem of computing quasi-upward planar drawings with minimum curve complexity of bimodal plane digraphs possibly having multiple edges. The \emph{curve complexity} is  the maximum number of bends along any edge of the drawing. We recall that minimizing the curve complexity is a classical subject of investigation in Graph Drawing (see, e.g.,~\cite{DBLP:journals/jgaa/BekosKK17,DBLP:journals/comgeo/BiedlK98,DBLP:conf/walcom/BinucciD14,DBLP:journals/comgeo/ChaplickLWZ19,DBLP:conf/gd/GiacomoDLM09,DBLP:journals/tcs/GiacomoGLN20,DBLP:journals/algorithmica/GiacomoLT10,DBLP:journals/algorithmica/Kant96,DBLP:journals/jgaa/KaufmannW02,JGAA-547}).
Our results can be summarized as follows.

\begin{itemize}
	\item In~\Cref{se:combinatoric} we prove that every bimodal plane digraph admits an embedding-preserving quasi-upward planar drawing with curve complexity two. This is worst-case optimal, since the number of bends per edge in a quasi-upward planar drawing is an even number and not all bimodal plane digraphs are upward planar~\cite{dett-gd-99}. This result is the counterpart in the quasi-upward planar setting of a well-known result by Biedl and Kant who prove, in the orthogonal setting, that every plane graph of degree at most four admits an orthogonal drawing with curve complexity two~\cite{DBLP:journals/comgeo/BiedlK98}.
	
	
	\item In~\Cref{se:unit-flow} we study the problem of minimizing the curve complexity in a quasi-upward planar drawing. This problem can be modeled as a min-cost flow problem on a unit-capacity planar flow network. By exploiting a result of Karczmarz and Sankowski~\cite{DBLP:conf/esa/KarczmarzS19} we obtain an algorithm to compute embedding-preserving quasi-upward planar drawings that minimize the curve complexity and that have the minimum number of bends when no edge can be bent more than twice that runs in ${\tilde O}(m^{\frac{4}{3}})$ time, where $m$ is the number of edges of the input digraph. We recall that the problem of computing planar drawings that minimize the number of bends while keeping the curve complexity bounded by a constant has already been studied, for example in the context of orthogonal representations (see, e.g.,~\cite{DBLP:conf/soda/DidimoLOP20,DBLP:conf/gd/DidimoLP18,DBLP:journals/jgaa/RahmanNN99}).
	
	
	\item A quasi-upward planar drawing with minimum curve complexity may have linearly many bends in total. Thus, a natural question to ask is whether these many bends are sometimes necessary if we just minimize the number of bends independent of the curve complexity and, if so, what the curve complexity may be. In \Cref{se:lowerbound} we prove that for every $n \geq 39$ there exists a planar bimodal digraph with $n$ vertices whose bend-minimum quasi-upward planar drawings have at least $cn$ bends on a single edge, for a constant $c>0$. We  show that this bound holds even in the variable embedding setting. This result can be regarded as the counterpart in the quasi-upward planar setting of a result by Tamassia et al.~\cite{DBLP:conf/spdp/TamassiaTV91} showing a similar lower bound on the curve complexity of bend-minimum planar orthogonal representations.
\end{itemize}

Preliminaries are in \Cref{se:preliminaries}, while \Cref{se:conclusions} lists some open problems. Proofs marked with ($\star$) are omitted or sketched and can be found in 
\ifArxiv
the \mbox{appendix}. 
\else
\cite{DBLP:journals/corr/abs-xxxx-xxxxx}.
\fi

\section{Preliminaries}\label{se:preliminaries}

We consider multi-digraphs, that are directed graphs which can have multiple edges. For simplicity we shall call them \emph{digraphs}. We also assume the digraphs to be connected; indeed a digraph $G$ has a quasi-upward drawing if and only if each connected component of $G$ has a quasi-upward drawing. A vertex of a digraph $G$ without incoming (outgoing) edges is a \emph{source} (\emph{sink}) of $G$. A vertex that is not a source nor a sink is an \emph{internal vertex}.

A \emph{drawing} $\Gamma$ of a digraph $G=(V,E)$ is a mapping of the vertices of $V$ to points of the plane, and of the edges in $E$ to Jordan arcs connecting their corresponding endpoints but not passing through any other vertex. Drawing $\Gamma$ is \emph{planar} if any two edges can only meet at common endpoints. A digraph is \emph{planar} if it admits a planar drawing. A planar drawing of a planar digraph $G$ subdivides the plane into topologically connected regions, called \emph{faces}. The infinite region is the \emph{external face}. A \emph{planar embedding} $\mathcal E$ of $G$ is an equivalence class of planar drawings that define the same set of faces and have the same external face. A planar embedding of a connected digraph can be uniquely identified by the clockwise circular order of the edges around each vertex and by the external face. A \emph{plane digraph} $G$ is a planar digraph with a given planar embedding. The number of vertices encountered in a closed walk along the boundary of a face $f$ of $G$ is the \emph{degree} of $f$, denoted as $\delta(f)$. If $G$ is not biconnected, a vertex may be encountered more than once, thus contributing  more than once to the degree of the face. The \emph{dual digraph} of $G$ is a plane digraph with a vertex for each face of $G$ and an edge $e'$ between two faces for each edge $e$ of $G$ shared by the two faces. The edge $e'$ is oriented from the face to the left of $e$ to the face \mbox{to the right of $e$.}

\section{Subdivisions of Bimodal Plane Digraphs}\label{se:combinatoric}

In this section we show how to suitably subdivide the edges of a bimodal plane digraph so to obtain an upward plane digraph. We start by recalling the notions of large angles and upward consistent assignments~\cite{DBLP:journals/algorithmica/BertolazziBLM94}.

\medskip
\noindent \textbf{Bimodality and upward consistent assignments.} Let $G$ be a bimodal plane digraph. Let $f$ be a face of $G$, let $e_1$ and $e_2$ be two consecutive edges encountered in this order when walking counterclockwise along the boundary of $f$, and let $v$ be the vertex shared by $e_1$ and $e_2$; the pair $(e_1,e_2)$ is an \emph{angle} of $f$ at vertex $v$ (\Cref{fig:quasi-upward-c} highlights the angles of a face $f$). Notice that if $v$ has exactly one incident edge $e$, then $e_1=e_2=e$ and the pair $(e,e)$ is also an angle of $f$ at $v$.
Let $f$ be a face of $G$, let $v$ be a vertex of $f$, and let $(e_1,e_2)$ be an angle of $f$ at $v$. Angle $(e_1,e_2)$ is a \emph{source-switch} of $f$ if $e_1$ and $e_2$ are both outgoing edges for $v$; $(e_1,e_2)$ is a \emph{sink-switch} of $f$ if $e_1$ and $e_2$ are both incoming edges for $v$. An angle of $f$ that is neither a source-switch nor a sink-switch is a \emph{non-switch} of $f$. It is easy to observe that for any face $f$ the number of source-switches equals the number of sink-switches (in \Cref{fig:quasi-upward-c} the source- and sink-switches of $f$ are indicated). The number of source-switches in a face $f$ is denoted by $\mathcal A(f)$. The \emph{capacity} of $f$ is $\mathcal A(f)+1$ if $f$ is the external face and it is $\mathcal A(f)-1$ otherwise.

\begin{lemma}\emph{\cite{DBLP:journals/algorithmica/BertolazziBLM94}}\label{le:supply-capacity}
	Let $G$ be a bimodal plane digraph. The number of source and sink vertices of $G$ is equal to the sum of the capacities of the faces of $G$.
\end{lemma}

Let $G$ be an upward plane digraph and let $\Gamma$ be an embedding-preserving upward planar drawing of $G$. The angles of $G$ correspond to geometric angles in $\Gamma$. In particular, an angle $(e,e)$ of $G$ corresponds to a $2\pi$ angle in $\Gamma$. For each face $f$ of $G$, and for each source- or sink-switch $a$ of $f$, we assign a label $L$ to $a$, if $a$ is larger than $\pi$ in $\Gamma$; in this case we say that $a$ is a \emph{large angle}. \Cref{fig:quasi-upward-e} shows an embedding-preserving upward planar drawing of the graph in \Cref{fig:quasi-upward-c} with the angles larger that $\pi$ highlighted; the corresponding angles in \Cref{fig:quasi-upward-d} are labeled with an $L$.  We denote the number of $L$ labels on the angles of $f$ by $L(f)$. Also, if $v$ is a vertex of $G$, we denote by $L(v)$ the number of $L$ labels on all angles at vertex $v$. In~\cite{DBLP:journals/algorithmica/BertolazziBLM94} it is shown that $L(v) = 0$ if $v$ is an internal vertex, and $L(v)=1$ if $v$ is a source or a sink. 
Also, $L(f) = \mathcal A(f) + 1$ if $f$ is the external face and $L(f)=\mathcal A(f) - 1$ otherwise; in other words, the number of large angles inside each face is equal to its capacity (see, e.g., faces $f$ and $f'$ in \Cref{fig:quasi-upward-d}).

A bimodal plane digraph $G$ is upward planar if and only if it is acyclic and it admits an upward consistent assignment~\cite{DBLP:journals/algorithmica/BertolazziBLM94}. An \emph{upward consistent assignment} is an assignment of the source and sink vertices of $G$ to its faces such that: (i) Each source or sink $v$ is assigned to exactly one of its incident faces. (ii) For each face $f$, the number of source and sink vertices assigned to $f$ is equal to the capacity of $f$. Assigning a source or sink $v$ to a face $f$ corresponds to assigning an $L$ label to an angle that $v$ forms in $f$. See \Cref{fig:quasi-upward-d} for an example. 

\begin{figure}[t!]
	\centering
	\subfigure[]{
		\includegraphics[width=0.3\textwidth, page=5]{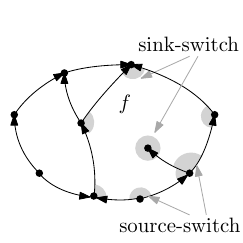}
		\label{fig:quasi-upward-c}
	}
	\hfil
	\subfigure[]{
		\includegraphics[width=0.3\textwidth, page=7]{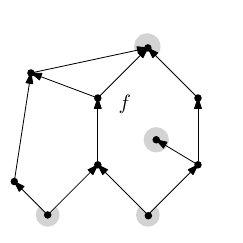}
		\label{fig:quasi-upward-e}
	}
	\hfil
	\subfigure[]{
		\includegraphics[width=0.3\textwidth, page=6]{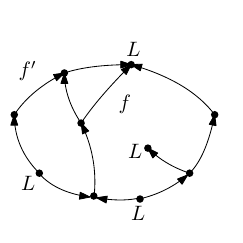}
		\label{fig:quasi-upward-d}
	}
	\caption{(a) Angles (shown in gray), source-switches, and sink-switches of a face $f$ of a bimodal plane digraph $G$; $\mathcal A(f) =2$; (b) An upward planar drawing of $G$. (c) The assignments of $L$ labels to the angles of $G$ ($L(f){=}1$, $L(f'){=}3$).}
	\label{fig:upward}
\end{figure}

\medskip
\noindent \textbf{$2$-subdivisions of bimodal plane digraphs.} Let $G$ be a bimodal plane digraph. A face $f$ of $G$ is \emph{nice} if $\delta(f)=4$ and each angle of $f$ is either a source-switch or a sink-switch (see, e.g., $f_2$ in \Cref{fig:quasi-triangulated-a}). We augment $G$ by adding edges (possibly creating multiple edges) so that the augmented digraph $G'$ is bimodal and each face of $G'$ either has degree two, or three, or it is nice (note that there can be more than one such augmentations). The resulting digraph is called a \emph{quasi-triangulation} of $G$ and all its faces have degree at most four. See \Cref{fig:quasi-triangulated-a,fig:quasi-triangulated-b}. The following lemma, whose proof is reported 
\ifArxiv
in the appendix for completeness, 
\else
in~\cite{DBLP:journals/corr/abs-xxxx-xxxxx} for completeness,
\fi
can also be derived as a \mbox{special case of~\cite[Lemma 5]{DBLP:journals/corr/abs-2008-07834}}.

\begin{figure}[t!]
	\centering
	\subfigure[]{
		\includegraphics[width=0.35\textwidth, page=1]{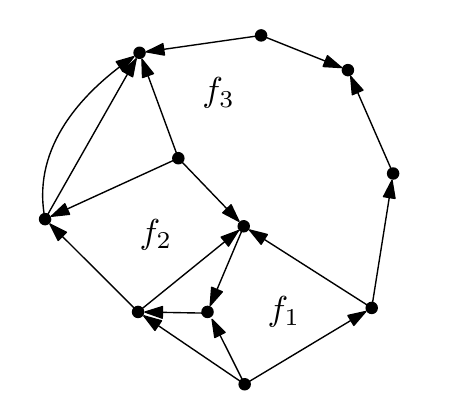}
		\label{fig:quasi-triangulated-a}
	}
	\hfil
	\subfigure[]{
		\includegraphics[width=0.35\textwidth, page=2]{figures/quasi-triangulated}
		\label{fig:quasi-triangulated-b}
	}
	\caption{(a) A bimodal plane digraph $G$; $f_2$ is nice, while $f_1$ is not; $f_3$ is such that $\delta(f_3) \geq 5$. (b) A quasi-triangulation of $G$: The added edges are colored blue.
	}
	\label{fig:quasi-triangulated.1}
\end{figure}

\begin{restatable}[$\star$]{lemma}{lequasitriangulation}\label{le:quasi-triangulation}
	Every bimodal plane digraph admits a quasi-triangulation.
\end{restatable}

The \emph{$2$-subdivision} of a bimodal plane digraph $G$ is the graph $\hat{G}$ obtained from $G$ by replacing each edge $e=(u,v)$ of $G$ with the three edges $(u,t_e)$, $(s_e,t_e)$, $(s_e,v)$, where $s_e$ and $t_e$ are two subdivision vertices. See \Cref{fig:bipartite-description-a} for an illustration. Notice that $s_e$ is a source and $t_e$ is a sink. Furthermore, the $2$-subdivision $\hat{G}$ of $G$ is bimodal, it has no multiple edges, and it is acyclic even if $G$ is not. We prove that $\hat{G}$ is upward planar, which implies that $G$ admits an embedding preserving quasi-upward planar drawing with curve complexity two. Since $\hat{G}$ is bimodal and acyclic, it is sufficient to prove that $\hat{G}$ admits an upward consistent assignment. We model the problem of computing an upward consistent assignment of $\hat{G}$ as a matching problem on a suitably defined bipartite~graph.


The \emph{bipartite description} of $\hat{G}$ is the bipartite graph $H_{\hat{G}}=(A,B,E)$ defined as follows. The vertex set $A$ contains for each face $\hat{f}$ of $\hat{G}$ a set of vertices $\ai{f}{1},\ai{f}{2},\dots,\ai{f}{c}$, where $c$ is the capacity of $\hat{f}$. Each vertex $\ai{f}{i}$ (for $i=1,\dots,c$) is called a \emph{representative vertex} of face $f$. The vertex set $B$ contains the source and sink vertices of $\hat{G}$. There is an edge $(\ai{f}{i},v)$ in $E$ if $v$ is a source or sink vertex of face $\hat{f}$ (for $i=1,\dots,c$). See \Cref{fig:bipartite-description-b} for an illustration.

\begin{restatable}[$\star$]{lemma}{leperfectmatching}\label{le:perfect-matching} The $2$-subdivision of a bimodal plane digraph admits an upward consistent assignment if and only if its bipartite description has a \mbox{perfect matching}.
\end{restatable}

A bimodal plane digraph is \emph{face-acyclic} if its face boundaries are not cycles. To prove the main result of this section, we first consider face-acyclic bimodal plane digraphs. We then show how to extend the result to the general case.

\begin{lemma}\label{le:hall}
The $2$-subdivision of a face-acyclic bimodal plane digraph is upward planar.
\end{lemma}
\begin{proof}
	Let $G$ be a face-acyclic bimodal plane digraph. We assume that $G$ is a quasi-triangulation. If not, by \Cref{le:quasi-triangulation} we can augment $G$ to a quasi-triangulation and the statement follows because the $2$-subdivision of $G$ is a subgraph of the $2$-subdivision of the obtained quasi-triangulation. By \Cref{le:perfect-matching}, it suffices to prove that $H_{\hat{G}}$ has a perfect matching. According to Hall's theorem $H_{\hat{G}}=(A,B,E)$ has a perfect matching if and only if for each $A' \subseteq A$, we have that $|A'| \leq |N(A')|$, where $N(A') \subseteq B$ is the set of neighbors of the vertices in $A'$~\cite{hall}.
	Let $A'$ be a subset of $A$ and let $\{f_1, \dots,f_k\}$ be the faces with a representative vertex in $A'$. $A'$ is \emph{complete} if it contains all representative vertices for each face $f_i$ ($1 \leq i \leq k$).

\begin{figure}[t!]
	\centering
	\subfigure[]{
		\includegraphics[width=0.35\textwidth, page=3]{figures/quasi-triangulated}
		\label{fig:bipartite-description-a}
	}
	\hfil
	\subfigure[]{
		\includegraphics[width=0.32\textwidth, page=1]{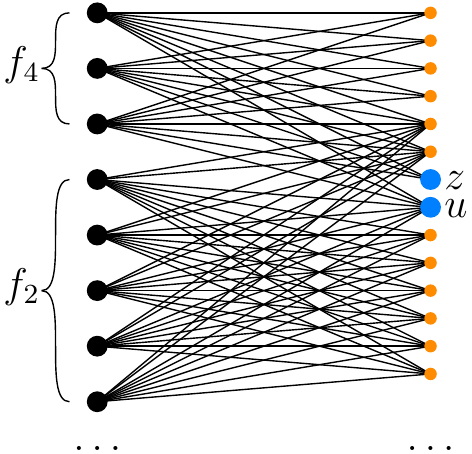}
		\label{fig:bipartite-description-b}
	}
	\caption{ (a) The $2$-subdivision $\hat{G}$ of the graph $G$ in \Cref{fig:quasi-triangulated-b}. (b) A portion of the bipartite description $H_{\hat{G}}$ of $\hat{G}$.
	}
	\label{fig:bipartite-description}
\end{figure}
	
	\begin{claimx}\label{cl:complete}
		Let $A'$ and $A''$ be two distinct subsets of $A$ that contain the representative vertices of the same set of faces. If $A'$ is complete and $|A'| \leq |N(A')|$, then  $|A''| \leq |N(A'')|$.
	\end{claimx}
	\noindent \emph{Proof.}
		Since $A'' \subseteq A'$ and since all the representative vertices of a face have the same neighbors in $B$, we have that $|A''| \leq |A'|$ and $N(A')=N(A'')$.\qedclaim
	
	\medskip


	By \Cref{cl:complete}, it is sufficient to prove that Hall's theorem holds for any complete subset $A'$ of $A$. Let $N_1(A') =\{v \in N(A')~|~v$ is a subdivision vertex of $\hat{G}\}$ and let $N_2(A')=N(A')\setminus N_1(A')$. Let $F$ be the set of the faces of $G$ and let $F' \subseteq F$
	be the set of faces whose representative vertices are in $A'$. We denote by $G_{F'}$ the bimodal plane subgraph of $G$ induced by the edges of the boundaries of the faces in $F'$. Note that $G_{F'}$ can have faces other than those in $F'$ (see \Cref{fig:quasi-triangulated-c,fig:quasi-triangulated-d}). Let $F''$ be the set of faces of $G_{F'}$ that are not in $F'$. Each face in $F''$ is the union of one or more faces of $F \setminus F'$. Further, let $F'_i=\{f \in F'~|~f \text{ has degree $i$ in $G$}\}$, for $i=2,3,4$. Let $E_b$ be the set of edges of $G$ shared by the faces in $F'$ and those in $F''$ (bold edges in \Cref{fig:quasi-triangulated-d}). Finally, we set $\alpha=1$ if the external face of $G$ belongs to $F'$, and $\alpha=0$ otherwise.
	
	\begin{claimx}\label{cl:n1-a}
		$|N_1(A')|-|A'|=|E_b|-|F'_4|-2\alpha$.
	\end{claimx}
	\noindent \emph{Proof.}
		For each edge of $G_{F'}$ there are two vertices in $N_1(A')$ because each edge has two subdivision vertices. Thus $|N_1(A')|=2|E_b|+2|E_x|$, where $E_x$ is the set of edges of $G_{F'}$ that are not in $E_b$.
	
\begin{figure}[t!]
		\centering
		\subfigure[]{
			\includegraphics[width=0.35\textwidth, page=4]{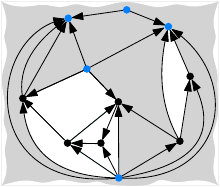}
			\label{fig:quasi-triangulated-c}
		}
		\hfil
		\subfigure[]{
			\includegraphics[width=0.35\textwidth, page=5]{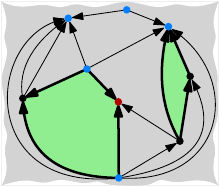}
			\label{fig:quasi-triangulated-d}
		}
		\caption{Illustration for \Cref{le:hall}. (a) Faces in $F'$ are gray.  (b) The graph $G_{F'}$; faces of $F''$ are green. The red vertex belongs to $S_{F'}$. The edges of $E_b$ are bold.}
		\label{fig:quasi-triangulated.2}
	\end{figure}

		Since $G$ is face-acyclic, each face of degree $2$ (resp. $3$ and $4$) in $G$ has capacity $2$ (resp. $3$ and $5$) in $\hat{G}$ if it is an internal face, and it has capacity $4$ (resp. $5$ and $7$) if it is the external face. It follows that for each face in $F'_2$ there are two vertices in $A'$, for each face in $F'_3$ there are three vertices in $A'$, and for each face in $F'_4$ there are five vertices in $A'$. If the external face of $G$ belongs to $F'$, then there are two additional vertices in $A'$.  Thus, $|A'|=2|F'_2|+3|F'_3|+5|F'_4|+2\alpha$ and $|N_1(A')|-|A'|=2|E_b|+2|E_x|-2|F'_2|-3|F'_3|-5|F'_4|-2\alpha$.
		
		Each face in $F'_i$ has $i$ edges in $G_{F'}$, for $i=2,3,4$. Each edge of $E_x$ belongs to two faces of $F'$, while each edge of $E_b$ belongs to only one face of $F'$. Thus we have $|E_b|+2|E_x|=2|F'_2|+3|F'_3|+4|F'_4|$ and therefore 
		$|N_1(A')|-|A'|=|E_b|+(2|F'_2|+3|F'_3|+4|F'_4|)-2|F'_2|-3|F'_3|-5|F'_4|-2\alpha=|E_b|-|F'_4|-2\alpha$.~\qedclaim
	
	\medskip
	
	Graph $G_{F'}$ can have some source or sink vertices that are not source or sink vertices of $G$ because $G_{F'}$ has a subset of the edges of $G$ (see for example the red vertex in \Cref{fig:quasi-triangulated-d}). Let $S_{F'}$ be the set of such vertices. Also, let $F''_2$ and $F''_3$ be the set of faces of $F''$ that have degree $2$ and $3$, respectively, in $G_{F'}$ and let $F''_x$ be the set $F'' \setminus (F''_3 \cup F''_2)$.
	
	\begin{claimx}\label{cl:n2}
		$|N_2(A')| \geq |F'_4|-\frac{|E_b|}{2}-\frac{|F''_3|}{2}-|F''_2|-|F''_x|+2$.
	\end{claimx}
	\noindent \emph{Proof.}
		By \Cref{le:supply-capacity}, the number of source and sink vertices of $G_{F'}$ is equal to the total capacity of the faces of $G_{F'}$. The number of source and sink vertices of $G_{F'}$ is $|N_2(A')| + |S_{F'}|$. The total capacity of the faces of $G_{F'}$ is given by two terms: The total capacity $C'$ of the faces in $F'$ plus the total capacity $C''$ of the faces in $F''$.  We have $C'=|F'_4|+2\alpha$. Indeed, let $f$ be a face of $F'$: If $f$ is internal it has capacity $0$ if it has degree $2$ or $3$ and it has capacity $1$ if it has degree $4$; if $f$ is the external face of $G_{F'}$, it has capacity $2$ if it has degree $2$ or $3$ and it has capacity $3$ if it has degree $4$. The term $C''$ is equal to $\sum_{f \in F''}(\mathcal A(f)-1)+2(1-\alpha)$, where $2(1-\alpha)$ takes into account the fact that the external face of $G_{F'}$ may belong to $F''$. (Recall that the capacity of a face is $\mathcal A(f)-1$ if it is internal and $\mathcal A(f)+1$ if it is external. Also, if $\alpha=1$ the external face of $G_{F'}$ belongs to $F'$, otherwise it belongs to $F''$.)
		
		Each vertex in $S_{F'}$ belongs to at least one face of $F''$. Let $v$ be a vertex of $S_{F'}$ and let $f$ be a face of $F''$ that contains $v$. Since $v$ is a source or a sink vertex, $f$ has at least one angle at $v$ that is either a source-switch or a sink-switch. In other words, for each vertex in $S_{F'}$ there is at least one source-switch or a sink-switch in a face of $F''$. Since at least $\frac{|S_{F'}|}{2}$ of these angles are source-switches or sink-switches, we have that $\sum_{f \in F''}\mathcal A(f) \geq \frac{|S_{F'}|}{2}$ (recall that $\mathcal A(f)$ is equal to the number of source-switches of $f$, which is equal to the number of sink-switches of $f$). It follows that $C''=\sum_{f \in F''}(\mathcal A(f)-1)+2(1-\alpha) \geq \frac{|S_{F'}|}{2}-|F''|+2(1-\alpha)$.
		
		Thus, we have $|N_2(A')| + |S_{F'}|=C'+C''\geq |F'_4|+2\alpha+\frac{|S_{F'}|}{2}-|F''|+2(1-\alpha)$. From $F''=F''_2 \cup F''_3 \cup F''_x$, we obtain $|N_2(A')| \geq |F'_4|-\frac{|S_{F'}|}{2}-|F''_2|-|F''_3|-|F''_x|+2$.
		Since each edge of $E_b$ is shared by a face of $F'$ and a face of $F''$ and since each face of $F''$ has only edges of $E_b$ in its boundary, we have that $\sum_{f \in F''}\delta(f)=|E_b|$. Also, each face in $F''_3$ has at most two switches and therefore at least one of its vertices does not belong to $|S_{F'}|$. It follows that $|S_{F'}| \leq \sum_{f \in F''}\delta(f)-|F''_3|=|E_b| -|F''_3|$, and therefore $|N_2(A')| \geq |F'_4|-\frac{|E_b|}{2}+\frac{|F''_3|}{2}-|F''_2|-|F''_3|-|F''_x|+2=|F'_4|-\frac{|E_b|}{2}-\frac{|F''_3|}{2}-|F''_2|-|F''_x|+2$.~\qedclaim
	\medskip
	
	In order to prove that the condition of Hall's theorem holds for $A'$, we will show that $|N(A')|-|A'|\geq 0$. By \Cref{cl:n1-a,cl:n2}, $|N(A')|-|A'|=|N_1(A')|+|N_2(A')|-|A'|\geq |E_b|-|F'_4|-2 \alpha + |F'_4|-\frac{|E_b|}{2}-\frac{|F''_3|}{2}-|F''_2|-|F''_x|+2   = \frac{|E_b|}{2} - \frac{|F''_3|}{2} - |F''_2| - |F''_x| +2 -2\alpha$.
	
	\smallskip
	
	\noindent Since the faces in $F''$ do not share edges, we have that $|E_b| \geq 2|F''_2|+3|F''_3|+4|F''_x|$, and therefore $|N(A')|-|A'| \geq \frac{2|F''_2|}{2} + \frac{3|F''_3|}{2} + \frac{4|F''_x|}{2} - \frac{|F''_3|}{2} - |F''_2| - |F''_x| +2-2\alpha=|F''_3|+|F''_x|+2-2\alpha$. Since $\alpha$ is either $0$ or $1$, we have that $|N(A')|-|A'| \geq |F''_3|+|F''_x| \geq 0$ and the condition of Hall's theorem holds.\qed
\end{proof}

The next theorem extends \Cref{le:hall} to graphs that may not be face-acyclic.

\begin{restatable}{theorem}{thhall}\label{th:hall}
	The $2$-subdivision of a bimodal plane digraph is upward \mbox{planar}.
\end{restatable}
\begin{proof}
	Let $G$ be a bimodal plane digraph. If $G$ is face-acyclic the statement follows from \Cref{le:hall}. Otherwise, for every face $f$ of $G$ whose boundary is a cycle, we insert a source vertex $v$ inside $f$ and connect it to every vertex of the boundary of $f$. Let $G'$ be the resulting digraph. Clearly, $G'$ is a face-acyclic bimodal digraph and by \Cref{le:hall} the $2$-subdivision $\hat{G'}$ of $G'$ is upward planar. Since the $2$-subdivision of $G$ is a subgraph of $\hat{G'}$, the statement follows.\qed
\end{proof}

\section{Computing Minimum Curve Complexity Drawings}\label{se:unit-flow}

To efficiently compute quasi-upward planar drawings with minimum curve complexity, we define a variant of the flow network used by Bertolazzi et al.~\cite{DBLP:journals/algorithmica/BertolazziBD02}. Feasible flows in this network correspond to quasi-upward planar drawings. Intuitively, each unit of flow represents a large angle; large angles are produced by sources and sinks and are consumed by the faces.

Let $G$ be a bimodal plane digraph. The \emph{unit-capacity flow network of $G$}, denoted as $\N(G)$, is defined as follows (see \Cref{fig:flow-network}). For each edge $e$ of $\N(G)$, we denote by $\B{e}$, $\C{e}$ and $\f{e}$ the capacity, the cost, and the flow of \mbox{$e$, respectively}.

\begin{itemize}
	\item The nodes of $\N(G)$ are all the sources and sinks (\emph{vertex-nodes}), and all the faces (\emph{face-nodes}) of $G$.
	\item Each vertex-node $v$ of $\N(G)$ supplies a flow equal to $1$. This means that exactly one of the angles at $v$ must be large.
	\item Each face-node of $\N(G)$ that corresponds to a face $f$ demands a flow equal to the capacity of $f$. This means that $f$ must have a number of large angles equal to its capacity. If $f$ is an internal face and it is a directed cycle then $\mathcal A(f)=0$ and the capacity of $f$ is $-1$, that is, $f$ supplies a flow equal to $1$.
	\item For each source or sink $v$ that belongs to a face $f$, there is a \emph{vertex-to-face arc} $(v,f)$ in $\N(G)$ such that $\B{v,f}=1$ and $\C{v,f}=0$. Intuitively, a unit of flow on this arc means that $f$ has a large angle at $v$.
	\item For each edge $e$ of $G$ shared by two faces $f$ and $g$, there is a pair of \emph{face-to-face arcs} $(f,g)$ and $(g,f)$ in $\N(G)$ such that $\B{f,g}=\B{g,f}=1$ and $\C{f,g}=\C{g,f}=2$. Intuitively, a unit of flow on $(f,g)$ or on $(g,f)$ represents the insertion of two bends along $e$. This corresponds to two units of cost. The two arcs $(f,g)$ and $(g,f)$ are called the \emph{dual arcs} of edge $e$.
\end{itemize}


\begin{figure}[t!]
	\centering
	\subfigure[]{
		\includegraphics[width=0.3\textwidth, page=1]{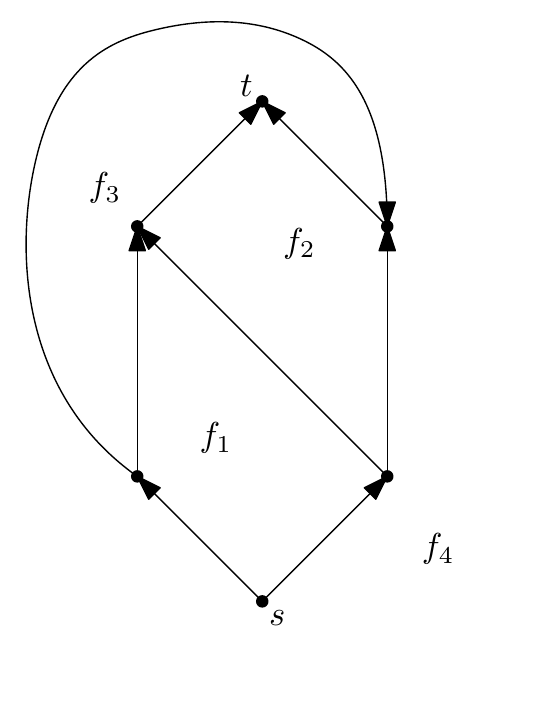}
		\label{fig:flow-network-a}
	}
	\hfil
	\subfigure[]{
		\includegraphics[width=0.3\textwidth, page=2]{figures/flow-network}
		\label{fig:flow-network-b}
	}
	\hfil
	\subfigure[]{
		\includegraphics[width=0.3\textwidth, page=3]{figures/flow-network}
		\label{fig:flow-network-c}
	}
	\caption{(a) A bimodal planar embedding of the graph $G$ of \Cref{fig:quasi-upward-a}. (b) The unit-capacity flow network $\N(G)$. The dashed arcs (vertex-to-face arcs) have cost 0, while solid arcs (face-to-face arcs) have cost $2$. The number close to each face-node indicates the capacity of the face.  (c) A feasible flow $\fl$ for $\N(G)$. A unit of flow traverses the edges highlighted in bold. The cost of the flow is $2$. A $2$-bend drawing of $G$ corresponding to $\fl$ is the one shown in \Cref{fig:quasi-upward-b}. }
	\label{fig:flow-network}
\end{figure}

We remark that the main differences between the flow network defined above and the flow network $\NB(G)$ defined by Bertolazzi et al.~\cite{DBLP:journals/algorithmica/BertolazziBD02} are as follows: \emph{(i)} In $\N(G)$ we have two opposite face-to-face arcs for each edge $e$ of $G$, while in $\NB(G)$ there are two opposite face-to-face arcs for each pair of adjacent faces, even when they share more than one edge; \emph{(ii)} the capacity of the face-to-face arcs is one in $\N(G)$ and it is unbounded in $\NB(G)$. The fact that $\N(G)$ is well-defined is a consequence of \Cref{le:supply-capacity}. The following lemma will be used to prove that a quasi-upward planar drawing with minimum curve complexity can be computed by means of the flow network $\N(G)$.

\begin{restatable}[$\star$]{lemma}{leflow}\label{le:flow}
	Let $G$ be a bimodal plane digraph. For each feasible flow $\fl_u$ in $\NB_u(G)$ there exists a quasi-upward planar drawing $\Gamma$ of $G$ such that the number of bends along each edge $e$ of $G$ is equal to the sum of the costs of the flows along the two face-to-face arcs that are the dual arcs of $e$.
\end{restatable}

The next theorem gives the main result of this section and exploits the algorithm by Karczmarz and Sankowski~\cite{DBLP:conf/esa/KarczmarzS19} to compute a min-cost flow on a planar unit-capacity flow network.

\begin{restatable}[$\star$]{theorem}{thmain}\label{th:main}
	Let $G$ be a bimodal plane digraph with $m$ edges. There exists an ${\tilde{O}}(m^\frac{4}{3})$-time algorithm that computes a quasi-upward planar drawing $\Gamma$ of $G$ with the following properties:
	\begin{inparaenum}[(i)]
		\item $\Gamma$ has minimum curve complexity, which is at most two.
		\item $\Gamma$ has the minimum number of bends among the quasi-upward planar drawings of $G$ with minimum curve complexity.
	\end{inparaenum}
\end{restatable}

If $G$ does not have homotopic multiple edges (i.e., multiple edges that define a face), then $m \in O(n)$ and the time complexity of \Cref{th:main} is ${\tilde{O}}(n^\frac{4}{3})$.

\section{A Lower Bound on the Curve Complexity}\label{se:lowerbound}

By \Cref{th:main}, every bimodal plane digraph admits a quasi-upward planar drawing with curve complexity two, which is worst-case optimal. One may wonder whether curve complexity two and minimum total number of bends can be simultaneously achieved. The next lemma shows that this is not always~possible.

\begin{figure}[t!]
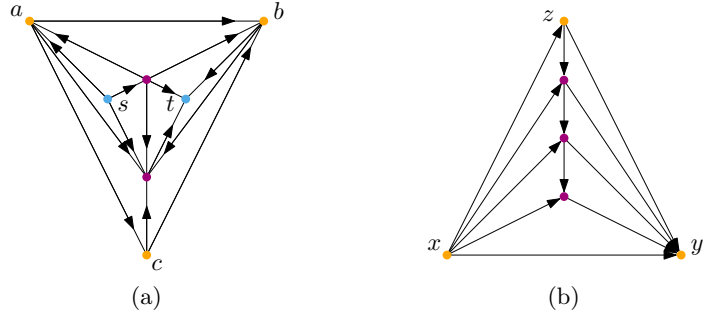

	\centering
	\subfigure[]{
		\includegraphics[width=0.32\textwidth, page=1]{figures/counterexample}
		\label{fig:counterexample-a}
	}
	\hfil
	\subfigure[]{
		\includegraphics[width=0.32\textwidth, page=2]{figures/counterexample}
		\label{fig:counterexample-b}
	}
	\caption{(a) Supplier gadget.  (b) Barrier gadget.}
	\label{fig:counterexample}
\end{figure}

\begin{figure}[p]
	\centering
	\includegraphics[width=\textwidth, page=3]{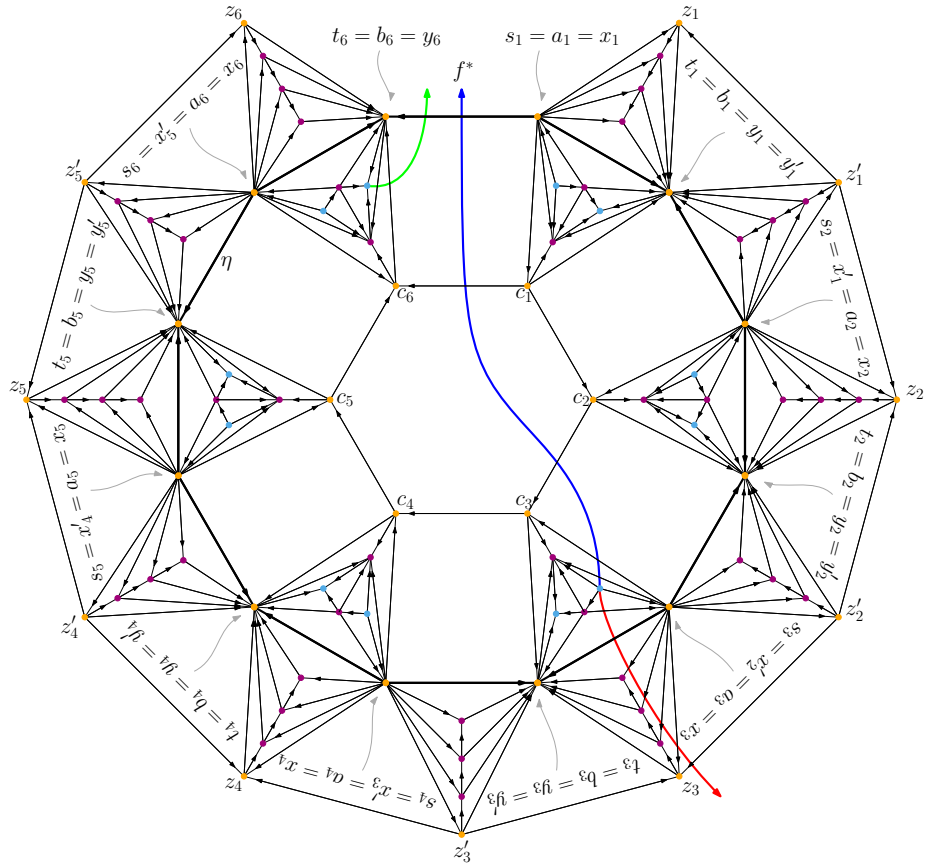}
	\caption{\label{fig:counterexample-c}Graph $G_6$ in the proof of \Cref{le:lowerbound}. The blue arrow represents a path in the dual that is shorter than the one represented by the red arrow. The green arrow also represents a shortest path.}
\end{figure}

\begin{restatable}[$\star$]{lemma}{lelowerbound}\label{le:lowerbound}
	For every integer $k \geq 3$, there exists a bimodal planar acyclic digraph $G_k$ with $14k-3$ vertices and $41k-13$ edges such that every quasi-upward planar drawing of $G_k$ with the minimum number of bends has one edge with at least $2k-2$ bends.
\end{restatable}
\begin{sketch}
For each $k \geq 3$, we construct a graph $G_k$ by suitably combining different copies of two gadgets. The first gadget is shown in \Cref{fig:counterexample-a} and it is called
\emph{supplier gadget} because it contains one source and one sink vertex, denoted as $s$ and $t$ in \Cref{fig:counterexample-a}, that supply two units of flow. Graph $G_k$ has $k$ copies of the supplier gadget that in total supply $2k$ units of flow; $G_k$ is such that $k-1$ units of this flow have to reach a specific face. To force these $k-1$ units of flow to ``traverse'' the same edge (thus creating $2k-2$ bends along this edge), we use the second gadget, shown in \Cref{fig:counterexample-b}. This gadget is called \emph{barrier gadget} because it is used to prevent the flow to traverse some edges. Graph $G_k$ is shown in \Cref{fig:counterexample-c} for $k=6$. 
\ifArxiv
See the appendix for more details on the~construction~of~$G_k$.
\else
See~\cite{DBLP:journals/corr/abs-xxxx-xxxxx} for more details on the~construction~of~$G_k$.
\fi


We now prove that every quasi-upward planar drawing of $G_k$ with the minimum number of bends has at least one edge with at least $2k-2$ bends. Let $\Gamma$ be any quasi-upward planar drawing of $G_k$ with the minimum number of bends and let $\psi$ be the planar embedding of $\Gamma$. Drawing $\Gamma$ corresponds to a minimum cost flow on the flow network $\NB(G_k)$ defined on the planar embedding $\psi$~\cite{DBLP:journals/algorithmica/BertolazziBD02}. Since $G_k$ is triconnected, all its planar embeddings have the same set of faces and each embedding is defined by the choice of a face as the external one. Let $f^*$ be the face that is external in the embedding of $G_k$ shown in \Cref{fig:counterexample-c}. Face $f^*$ has $k$ source-switches, which implies that its capacity is at least $k-1$ in $\NB(G_k)$: Namely, it is $k+1$ if $f^*$ is the external face in $\psi$ and $k-1$ otherwise.
Since there are one source and one sink in each of the $k$ supplier gadgets, the total amount of flow consumed by the faces is $2k$. In any planar embedding of $G_k$ the cycle $\eta$, highlighted by bold edges in \Cref{fig:counterexample-c}, separates the source and sink vertices from $f^*$. It follows that at least $k-1$ units of flow must go through the dual arcs of the edges of $\eta$, thus creating at least $2k-2$ bends along the edges of $\eta$. We now prove that all these bends are on the edge $(s_1,t_k)$.
Consider a unit flow that goes from a source or a sink node $v$ to $f^*$ following a path $\pi$ in $\NB(G_k)$; the cost of sending this unit of flow along $\pi$ is equal to the number of face-to-face arcs in  $\pi$ because face-to-face arcs have cost two. Each face-to-face arc of $\pi$ is the dual arc of an edge in $G_k$ shared by two adjacent faces. Hence, to obtain a minimum cost flow, each unit of flow that goes from a source or a sink node $v$ to $f^*$ in $\NB(G_k)$ follows a shortest path $\pi$ in the dual graph of $G_k$ connecting a face incident to $v$ with $f^*$. The lemma holds because any shortest path in the dual graph of $G_k$ connecting a face incident to $v$ with $f^*$ includes a dual arc of the edge $(s_1,t_k)$. 
\ifArxiv
(See the appendix for more details.)~\hfill$\qed$\endproof
\else
(See~\cite{DBLP:journals/corr/abs-xxxx-xxxxx} for more details.)~\hfill$\qed$\endproof
\fi
\end{sketch}

The next theorem extends the result of \Cref{le:lowerbound} to the cases when $n$ is not equal to $14k-3$, for some $k>0$.

\begin{restatable}[$\star$]{theorem}{thlower}\label{th:lower}
	For every $n \geq 39$ there exists a bimodal planar digraph with $n$ vertices such that every bend-minimum quasi-upward planar drawing has curve complexity at least $\frac{n-24}{7}$.
\end{restatable}

\section{Open Problems}\label{se:conclusions}


We conclude by mentioning some open problems that are naturally suggested by the results in this paper.
\begin{itemize}
	\item Is it possible to improve the time complexity stated by \Cref{th:main}?
	\item We showed that every bimodal plane digraph becomes upward planar if every edge is subdivided twice. It would be interesting to minimize the total number of subdivision vertices (with at most two subdivision vertices per edge) such that the resulting graph admits an upward straight-line drawing of polynomial area.
\end{itemize}

%

\bibliography{biblio}

\begin{thebibliography}{10}
\providecommand{\url}[1]{\texttt{#1}}
\providecommand{\urlprefix}{URL }
\providecommand{\doi}[1]{https://doi.org/#1}

\bibitem{DBLP:journals/corr/abs-2008-07834}
Angelini, P., Chaplick, S., Cornelsen, S., {Da Lozzo}, G.: Planar {L}-drawings
  of bimodal graphs. CoRR  \textbf{abs/2008.07834} (2020),
  \url{https://arxiv.org/abs/2008.07834}

\bibitem{DBLP:journals/jgaa/BekosKK17}
Bekos, M.A., Kaufmann, M., Krug, R.: On the total number of bends for planar
  octilinear drawings. J. Graph Algorithms Appl.  \textbf{21}(4),  709--730
  (2017)

\bibitem{DBLP:journals/algorithmica/BertolazziBD02}
Bertolazzi, P., {Di Battista}, G., Didimo, W.: Quasi-upward planarity.
  Algorithmica  \textbf{32}(3),  474--506 (2002).
  \doi{10.1007/s00453-001-0083-x}

\bibitem{DBLP:journals/algorithmica/BertolazziBLM94}
Bertolazzi, P., {Di Battista}, G., Liotta, G., Mannino, C.: Upward drawings of
  triconnected digraphs. Algorithmica  \textbf{12}(6),  476--497 (1994).
  \doi{10.1007/BF01188716}

\bibitem{DBLP:journals/comgeo/BiedlK98}
Biedl, T.C., Kant, G.: A better heuristic for orthogonal graph drawings.
  Comput. Geom.  \textbf{9}(3),  159--180 (1998).
  \doi{10.1016/S0925-7721(97)00026-6}

\bibitem{DBLP:journals/corr/abs-xxxx-xxxxx}
Binucci, C., {Di Giacomo}, E., Liotta, G., Tappini, A.: Quasi-upward planar
  drawings with minimum curve complexity. CoRR  \textbf{abs/2108.10784} (2021),
  \url{https://arxiv.org/abs/2108.10784}

\bibitem{DBLP:conf/walcom/BinucciD14}
Binucci, C., Didimo, W.: Quasi-upward planar drawings of mixed graphs with few
  bends: Heuristics and exact methods. In: {WALCOM}. Lecture Notes in Computer
  Science, vol.~8344, pp. 298--309. Springer (2014)

\bibitem{DBLP:journals/cj/BinucciD16}
Binucci, C., Didimo, W.: Computing quasi-upward planar drawings of mixed
  graphs. Comput. J.  \textbf{59}(1),  133--150 (2016).
  \doi{10.1093/comjnl/bxv082}

\bibitem{DBLP:journals/tcs/BinucciDP14}
Binucci, C., Didimo, W., Patrignani, M.: Upward and quasi-upward planarity
  testing of embedded mixed graphs. Theor. Comput. Sci.  \textbf{526},  75--89
  (2014). \doi{10.1016/j.tcs.2014.01.015}

\bibitem{10.1007/978-3-540-30140-0_16}
Chan, H.: A parameterized algorithm for upward planarity testing. In: Albers,
  S., Radzik, T. (eds.) Algorithms -- ESA 2004. pp. 157--168. Springer Berlin
  Heidelberg, Berlin, Heidelberg (2004)

\bibitem{DBLP:journals/comgeo/ChaplickLWZ19}
Chaplick, S., Lipp, F., Wolff, A., Zink, J.: Compact drawings of 1-planar
  graphs with right-angle crossings and few bends. Comput. Geom.  \textbf{84},
  50--68 (2019)

\bibitem{dett-gd-99}
{Di Battista}, G., Eades, P., Tamassia, R., Tollis, I.G.: Graph Drawing.
  Prentice Hall, Upper Saddle River, NJ (1999)

\bibitem{DBLP:journals/dcg/BattistaTT92}
{Di Battista}, G., Tamassia, R., Tollis, I.G.: Area requirement and symmetry
  display of planar upward drawings. Discret. Comput. Geom.  \textbf{7},
  381--401 (1992). \doi{10.1007/BF02187850}

\bibitem{DBLP:conf/gd/GiacomoDLM09}
{Di Giacomo}, E., Didimo, W., Liotta, G., Meijer, H.: Area, curve complexity,
  and crossing resolution of non-planar graph drawings. In: Graph Drawing.
  Lecture Notes in Computer Science, vol.~5849, pp. 15--20. Springer (2009)

\bibitem{DBLP:journals/tcs/GiacomoGLN20}
{Di Giacomo}, E., Gasieniec, L., Liotta, G., Navarra, A.: On the curve
  complexity of 3-colored point-set embeddings. Theor. Comput. Sci.
  \textbf{846},  114--140 (2020)

\bibitem{DBLP:journals/algorithmica/GiacomoLT10}
{Di Giacomo}, E., Liotta, G., Trotta, F.: Drawing colored graphs with
  constrained vertex positions and few bends per edge. Algorithmica
  \textbf{57}(4),  796--818 (2010). \doi{10.1007/s00453-008-9255-2}

\bibitem{DBLP:reference/algo/Didimo16}
Didimo, W.: Upward graph drawing. In: Encyclopedia of Algorithms, pp.
  2308--2312 (2016). \doi{10.1007/978-1-4939-2864-4\_653}

\bibitem{DBLP:journals/siamdm/DidimoGL09}
Didimo, W., Giordano, F., Liotta, G.: Upward spirality and upward planarity
  testing. {SIAM} J. Discret. Math.  \textbf{23}(4),  1842--1899 (2009).
  \doi{10.1137/070696854}

\bibitem{DBLP:conf/soda/DidimoLOP20}
Didimo, W., Liotta, G., Ortali, G., Patrignani, M.: Optimal orthogonal drawings
  of planar 3-graphs in linear time. In: Chawla, S. (ed.) Proceedings of the
  2020 {ACM-SIAM} Symposium on Discrete Algorithms, {SODA} 2020. pp. 806--825.
  {SIAM} (2020). \doi{10.1137/1.9781611975994.49}

\bibitem{DBLP:conf/gd/DidimoLP18}
Didimo, W., Liotta, G., Patrignani, M.: Bend-minimum orthogonal drawings in
  quadratic time. In: Biedl, T.C., Kerren, A. (eds.) Graph Drawing and Network
  Visualization - 26th International Symposium, {GD} 2018. Lecture Notes in
  Computer Science, vol. 11282, pp. 481--494. Springer (2018).
  \doi{10.1007/978-3-030-04414-5\_34}

\bibitem{DBLP:journals/siamcomp/GargT01}
Garg, A., Tamassia, R.: On the computational complexity of upward and
  rectilinear planarity testing. {SIAM} J. Comput.  \textbf{31}(2),  601--625
  (2001). \doi{10.1137/S0097539794277123}

\bibitem{hall}
Hall, P.: On representatives of subsets. J. London Math. Soc.
  \textbf{s1-10}(1),  26--30 (1935)

\bibitem{DBLP:journals/ijfcs/HealyL06}
Healy, P., Lynch, K.: Two fixed-parameter tractable algorithms for testing
  upward planarity. Int. J. Found. Comput. Sci.  \textbf{17}(5),  1095--1114
  (2006). \doi{10.1142/S0129054106004285}

\bibitem{DBLP:journals/algorithmica/Kant96}
Kant, G.: Drawing planar graphs using the canonical ordering. Algorithmica
  \textbf{16}(1),  4--32 (1996). \doi{10.1007/BF02086606}

\bibitem{DBLP:conf/esa/KarczmarzS19}
Karczmarz, A., Sankowski, P.: Min-cost flow in unit-capacity planar graphs. In:
  Bender, M.A., Svensson, O., Herman, G. (eds.) 27th Annual European Symposium
  on Algorithms, {ESA} 2019. LIPIcs, vol.~144, pp. 66:1--66:17. Schloss
  Dagstuhl - Leibniz-Zentrum f{\"{u}}r Informatik (2019).
  \doi{10.4230/LIPIcs.ESA.2019.66}

\bibitem{DBLP:journals/jgaa/KaufmannW02}
Kaufmann, M., Wiese, R.: Embedding vertices at points: Few bends suffice for
  planar graphs. J. Graph Algorithms Appl.  \textbf{6}(1),  115--129 (2002).
  \doi{10.7155/jgaa.00046}

\bibitem{JGAA-547}
{Kindermann}, P., {Montecchiani}, F., {Schlipf}, L., {Schulz}, A.: Drawing
  subcubic 1-planar graphs with few bends, few slopes, and large angles.
  Journal of Graph Algorithms and Applications  \textbf{25}(1),  1--28 (2021).
  \doi{10.7155/jgaa.00547}

\bibitem{DBLP:journals/jgaa/RahmanNN99}
Rahman, M.S., Nakano, S., Nishizeki, T.: A linear algorithm for bend-optimal
  orthogonal drawings of triconnected cubic plane graphs. J. Graph Algorithms
  Appl.  \textbf{3}(4),  31--62 (1999). \doi{10.7155/jgaa.00017}

\bibitem{DBLP:conf/spdp/TamassiaTV91}
Tamassia, R., Tollis, I.G., Vitter, J.S.: Lower bounds and parallel algorithms
  for planar orthogonal grid drawings. In: Proceedings of the Third {IEEE}
  Symposium on Parallel and Distributed Processing, {SPDP} 1991, 2-5 December
  1991, Dallas, Texas, {USA}. pp. 386--393. {IEEE} Computer Society (1991).
  \doi{10.1109/SPDP.1991.218215}

\end{thebibliography}
\bibliographystyle{splncs04}

\ifArxiv

\newpage

\appendix

\section*{Appendix}

\section{Additional Material for \Cref{se:combinatoric}}

\lequasitriangulation*
\begin{proof}
	Let $G$ be a bimodal plane digraph.	We prove that until $G$ has a face $f$ whose degree is larger than $4$ or it is not nice, we can add an edge inside $f$ without violating the bimodality.
	If $f$ is a face of degree four that is not nice, then it has an angle $(e_1,e_2)$ at a vertex $w$ that is a non-switch. In this case, it is possible to add an edge between the end-vertices of $e_1$ and $e_2$ different from $w$ without violating the bimodality (see $f_1$ in \Cref{fig:quasi-triangulated-b}).
	Let $f$ be a face of $G$ such that $\delta(f) \geq 5$. If $f$ has an angle $(e_1,e_2)$ at a vertex $w$ that is a non-switch, we can proceed as in the previous case. If all angles of $f$ are source- and sink-switches, since $\delta(f) \geq 5$ there are a source-switch at a vertex $u$ and a sink-switch at a vertex $v$ such that $u \neq v$ and $u$ and $v$ are not adjacent in the boundary of $f$. It follows that it is possible to add edge $(u,v)$ inside $f$ (see $f_3$ in \Cref{fig:quasi-triangulated-b}).
	\qed
\end{proof}


\leperfectmatching*
\begin{proof}
	Let $G$ be a bimodal plane digraph.	Suppose first that $G$ admits an upward consistent assignment. According to this assignment each source or sink vertex is assigned to a face and the number of source and sink vertices assigned to a face is equal to its capacity. Let $v_1,v_2,\dots,v_c$ be the source and sink vertices assigned to a face $f$ (where $c$ is the capacity of $f$). We select the edges $(\ai{f}{i},v_i)$ of $H_{\hat{G}}$ (for $i=1,2,\dots,c$). Since the selected edges are independent, they are part of a matching, and since exactly one edge is selected for each vertex $\ai{f}{i}$, the matching is perfect.
	
	Suppose now that $H_{\hat{G}}$ has a perfect matching. For each face $f$ of $G$ with capacity $c$, let $(\ai{f}{i},v_i)$ (for $i=1,2,\dots,c$) be the edges of $H_{\hat{G}}$ that belong to the perfect matching. We assign vertices $v_1,v_2,\dots,v_c$ to face $f$. Clearly, each source or sink vertex of $G$ is assigned to a face of $G$ and the number of source and sink vertices that are assigned to each face of $G$ is equal to its capacity.\qed
\end{proof}


\section{Additional Material for \Cref{se:unit-flow}}

\leflow*
\begin{proof}
	The flow network $\NB(G)$ is equivalent to a flow network $\NB^+(G)$ obtained as follows. Let $f$ and $g$ be two faces of $G$ that share edges $e_1, e_2,\dots, e_k$, for $k \geq 1$. The face-to-face arc $(f,g)$ of $\NB(G)$ is replaced in $\NB^+(G)$ by a set of face-to-face arcs $a_1, a_2,\dots,a_k$ such that each $a_i$ has unbounded capacity and unit cost. Such a network coincides with $\N(G)$ where all edges have unbounded capacity. Clearly, any feasible flow $\fl_u$ of $\N(G)$ is also a feasible flow of $\NB^+(G)$. As a consequence of~\cite[Lemma 7]{DBLP:journals/algorithmica/BertolazziBD02}, a feasible flow in $\NB^+(G)$ corresponds to a quasi-upward planar drawing such that the number of bends along each edge $e$ of $G$ is equal to the sum of the costs of the flows along the two dual arcs of $e$.\qed  
\end{proof}

\thmain*

\begin{proof}
	To construct the drawing $\Gamma$ we compute a min-cost flow on $\N(G)$. Denote by $n$, $m$, and $f$, the number of vertices, edges and faces of $G$, respectively. Since $G$ is connected and planar, $n \in O(m)$ and $f \in O(m)$. (We remark that although $G$ is planar, in general it is not true that $m \in O(n)$ because $G$ has multiple edges.) Thus, $\N(G)$ has $O(m)$ nodes and arcs, and it can be constructed in $O(m)$ time.  
	Since $\N(G)$ is planar and each edge has unit capacity, we can use the algorithm by Karczmarz and Sankowski~\cite{DBLP:conf/esa/KarczmarzS19}, whose time complexity is ${\tilde{O}}((NM)^{\frac{2}{3}}\log C)$, where $N$ and $M$ are the number of vertices and edges of the flow network, respectively, and $C$ is an upper bound to the edge costs. For our flow network $\N(G)$, we have $N \in O(m)$, $M \in O(m)$ and $C \in O(1)$. Thus, we can compute a min-cost flow on $\N(G)$ in ${\tilde{O}}(m^{\frac{4}{3}})$ time. Once a min-cost flow is computed, a drawing $\Gamma$ of $G$ can be constructed in $O(m+b)$ time~\cite{dett-gd-99}, where $b$ is the total number of bends. We now show that the curve complexity of $\Gamma$ is two and therefore $b \in O(m)$.
	
	By \Cref{le:flow}, the computed flow defines a drawing $\Gamma$ such that the number of bends along an edge $e$ of $G$ is equal to the sum of the costs of the flows along the two face-to-face arcs that are dual of $e$. Since the capacity of the arcs is one, then each arc has at most one unit of flow; since the flow has minimum cost, the dual arcs of an edge $e$ cannot have both one unit of flow. It follows that each edge of $G$ has at most two bends. Moreover, if $G$ is upward planar (i.e., it can be drawn without bends), then the face-to-face arcs have zero unit of flow, because they are the only arcs having an associated cost. So the curve complexity of the computed drawing is minimum. The fact that the flow has minimum cost implies  that $\Gamma$ has the minimum number of bends among all quasi-upward planar drawings with minimum curve complexity. 
	\qed
\end{proof}

\section{Additional Material for \Cref{se:lowerbound}}

\lelowerbound*

\begin{proof}
	For each $k \geq 3$, we construct a graph $G_k$ by suitably combining different copies of two gadgets. The first gadget is shown in \Cref{fig:counterexample-a} and it is called
	\emph{supplier gadget} because it contains one source and one sink vertex, denoted as $s$ and $t$ in \Cref{fig:counterexample-a}, that supply two units of flow. The edge connecting the two vertices denoted as $a$ and $b$ in \Cref{fig:counterexample-a} is called the \emph{base} of the supplier gadget, while the vertex denoted as $c$ is called the \emph{apex} of the supplier gadget. Graph $G_k$ has $k$ copies of the supplier gadget that in total supply a $2k$ units of flow; $G_k$ is such that at least $k-1$ units of this flow have to reach a specific face. To force these $k-1$ units of flow to ``traverse'' the same edge (thus creating $2k-2$ bends along this edge) we use the second gadget, which is shown in \Cref{fig:counterexample-b}. This gadget is called \emph{barrier gadget} because it is used to prevent the flow to traverse some edges. Similar to the supplier gadget, the edge connecting the two vertices denoted as $x$ and $y$ in \Cref{fig:counterexample-b} is called the \emph{base} of the barrier gadget, while the vertex denoted as $z$ is called the \emph{apex} of the barrier gadget.
	

	More precisely, $G_k$ is constructed as follows. See \Cref{fig:counterexample-c} for an illustration. Let $\eta$ be a cycle consisting of $2k$ vertices $s_1$, $t_1$, $s_2$, $t_2$, $\dots$, $s_k$, $t_k$, in this cyclic order (highlighted by bold edges in \Cref{fig:counterexample-c}). The edges of the cycle are oriented so that each $s_i$ is a source and each $t_i$ is a sink.
	The graph $G_k$ contains $k$ copies $S_1,S_2,\dots,S_k$ of the supplier gadget. These copies are arranged so that the base $(a_i,b_i)$ of gadget $S_i$ coincides with edge $(s_i,t_i)$ of cycle $\eta$, for $i=1,2,\dots,k$, and the gadget is embedded inside the cycle $\eta$. The apices $c_1,c_2,\dots,c_k$ of the supplier gadgets are connected with edges $(c_i,c_{i+1})$, for $i=1,2,\dots,k-1$, and edge $(c_1,c_k)$; notice that edge $(c_1,c_k)$ is oppositely oriented with respect to the other edges. The graph $G_k$ has $2k-1$ copies of the barrier gadget; $k$ of these gadgets are denoted $B_1, B_2, \dots, B_{k}$, while the other $k-1$ are denoted as $B'_1, B'_2, \dots, B'_{k-1}$. These copies are arranged so that: (i) the base $(x_i,y_i)$ of the gadget $B_i$ coincides with edge $(s_i,t_i)$ of $\eta$, for $i=1,2,\dots,k$, and the gadget is embedded outside the cycle $\eta$; (ii) the base $(x'_i,y'_i)$ of the gadget $B'_i$ coincides with edge $(s_{i+1},t_i)$ of $\eta$, for $i=1,2,\dots,k-1$, and the gadget is embedded outside the cycle $\eta$. Notice that the gadgets in the second set are mirrored with respect to the ones of the first set (i.e., their bases are oppositely oriented). The apices of these gadgets are connected with edges $(z'_i,z_i)$ and $(z'_i,z_{i+1})$, for $i=1,2,\dots,k-1$. Note that the graph $G_k$ is acyclic.
	
	
	We now prove that every quasi-upward planar drawing of $G_k$ with the minimum number of bends has at least one edge with at least $2k-2$ bends. Let $\Gamma$ be any quasi-upward planar drawing of $G_k$ with the minimum number of bends and let $\psi$ be the planar embedding of $\Gamma$. Drawing $\Gamma$ corresponds to a minimum cost flow on the flow network $\NB(G_k)$ defined on the planar embedding $\psi$~\cite{DBLP:journals/algorithmica/BertolazziBD02}.
	Since $G_k$ is triconnected all its planar embeddings have the same set of faces and each embedding is defined by the choice of a face as the external one. Let $f^*$ be the face that is external in the embedding of $G_k$ shown in \Cref{fig:counterexample-c}. Face $f^*$ has $k$ source-switches, which implies that its capacity is at least $k-1$ in $\NB(G_k)$: Namely, it is $k+1$ if $f^*$ is the external face in $\psi$ and $k-1$ otherwise.
	Since there are one source and one sink in each of the $k$ supplier gadgets, the total amount of flow consumed by the faces is $2k$. In any planar embedding of $G_k$ the cycle $\eta$, highlighted by bold edges in \Cref{fig:counterexample-c}, separates the source and sink vertices from $f^*$.  It follows that at least $k-1$ units of flow must go through the dual arcs of the edges of $\eta$, thus creating at least $2k-2$ bends along the edges of $\eta$. We now prove that all these bends are on the edge $(s_1,t_k)$. Consider a unit flow that goes from a source or a sink node $v$ to $f^*$ following a path $\pi$ in $\NB(G_k)$; the cost of sending this unit of flow along $\pi$ is equal to twice the number of face-to-face arcs in  $\pi$ because face-to-face arcs have cost two. Each face-to-face arc of $\pi$ is the dual arc of an edge in $G_k$ shared by two adjacent faces. It follows that in order to obtain a minimum cost flow, each unit of flow that goes from a source or a sink node $v$ to $f^*$ in $\NB(G_k)$ follows a shortest path in the dual graph of $G_k$ connecting a face incident to $v$ with $f^*$. If $\pi$ in the dual graph of $G_k$ includes a dual arc of an edge $e$, we say that $\pi$ \emph{traverses} $e$. We now show that any shortest path in the dual graph of $G_k$ connecting a face incident to $v$ with $f^*$ traverses edge $(s_1,t_k)$.
	
	Let $v$ be a source or sink node of a supplier gadget $S_i$ different from the source of $S_1$ and from the sink of $S_k$. There exists a path of length $5$ from a face incident to $v$ to $f^*$ that traverses $(s_1,t_k)$ (see the blue path in \Cref{fig:counterexample-c}); such a path traverses two edges of $S_i$, either edge $(c_{i-1},c_i)$ or edge $(c_i,c_{i+1})$, edge $(c_1,c_k)$, and edge $(s_1,t_k)$. If a path from a face incident to $v$ to $f^*$ exits $\eta$ by traversing an edge $e$ different from $(s_1,t_k)$, then it must traverse the barrier gadget whose base coincides with $e$, and thus its length is greater than $5$ (see the red path in \Cref{fig:counterexample-c}).
	Consider now the source of $S_1$ and the sink of $S_k$. In this case, there exists a path of length $3$ from a face incident to these vertices to $f^*$ that traverses $(s_1,t_k)$ (see the green path in \Cref{fig:counterexample-c}); this path traverses two edges of $S_i$ ($i \in \{1,k\}$) and edge $(s_1,t_k)$. Also in this case, any path that exits $\eta$ by traversing an edge different from $(s_1,t_k)$ must traverse a barrier gadget, and thus its length is greater than $3$. From the discussion above, it follows that a flow of minimum cost in $\NB(G_k)$ has at least $k-1$ units of flow that traverse $(s_1,t_k)$.
	%
	\qed
\end{proof}

\thlower*

\begin{proof}
	Given an integer $n \geq 39$, let $k=\left \lfloor \frac{n+3}{14} \right \rfloor \geq 3$. We construct a bimodal planar graph $G$ starting from the graph $G_k$ defined in the proof of \Cref{le:lowerbound} and by adding to it $h=n-(14k-3)$ vertices connecting each of them to $z'_1$ with an edge outgoing from $z'_1$. Notice that $h \leq 13$ and that these vertices are all sinks. It follows that each of them increases by one unit the capacity of the face inside which it is embedded, and it also supplies one unit of flow. Since the length of the shortest path from a source or sink vertex of a supplier gadget to face $f^*$ is smaller that then length of the shortest path from the same source or sink vertex to each of the faces where these additional vertices can be embedded, the argument used in the proof of \Cref{le:lowerbound} still holds. It follows that $G$ has curve complexity at least $2k-2=2(\frac{n-h+3}{14})-2=\frac{n-h-11}{7}$. Since $h \leq 13$, the curve complexity is at least $\frac{n-24}{7}$.\qed
\end{proof}

\fi

\end{document}